%% file: main.tex
\documentclass%[10pt,twocolumn]
{article} 

\input{packages.tex}

\input{macros.tex}

\begin{document}

\title{Modular Construction of Fixed Point Combinators and Clocked B\"{o}hm Trees}

\author{%
  J\"{o}rg Endrullis \quad\quad Dimitri Hendriks \quad\quad Jan Willem Klop\\[1ex]
  Vrije Universiteit Amsterdam,
  Department of Computer Science\\
  De Boelelaan 1081a,
  1081 HV Amsterdam,
  The Netherlands\\[.5ex]
  {\normalsize{\texttt{joerg@few.vu.nl \quad diem@cs.vu.nl \quad jwk@cs.vu.nl}}}
}

\date{}

%\keywords{fixed point combinators, lambda calculus, looping combinators, head reduction}
%\subjclass{F.4.1 [Mathematical logic and formal languages]: Lambda calculus and related systems}

\maketitle
\thispagestyle{empty}

\begin{abstract}
  \input{abstract}
\end{abstract}

\begin{quotation}
  \noindent
  The theory of sage birds (technically called \emph{fixed point combinators}) %\\
  is a fascinating and basic part of combinatory logic; %\\
  we have only scratched the surface.
  \ \\[0.75ex]
  \hspace*{\fill} R.~Smullyan~\cite{smull:1985}.
\end{quotation}

\section{Introduction}
\input{intro}

\section{Preliminaries}
\input{prelims}

%\section{How to Make Fixed Point Combinators}
%\input{making-fpcs}

\section{The \boehm{} Sequence}
\input{boehm}

\section{The Scott Sequence}
\input{scott}

\section{Generalized Generation Schemes}
\input{schemes}

\section{Clock Behaviour of Lambda Terms}\label{sec:clocked}
\input{clocked}

\section{Atomic Clocks}\label{sec:atom}
\input{atomic}

\section{Clocked \levi{} and \ber{} Trees}\label{sec:levi}
\input{levi}

\section{Concluding remarks}
\input{conclusion}

\bibliographystyle{plain}

\bibliography{main}

\end{document}

%% file: packages.tex
\usepackage{latexsym}
\usepackage{amsmath,amssymb,amstext,amsthm}
\usepackage{wasysym}

\usepackage{graphicx}
\usepackage{pstricks, pst-plot, pst-node, pst-tree, pst-text}
\usepackage{wrapfig}
\usepackage{sidecap}

\usepackage{comment}
\usepackage{tabularx}
\usepackage{stmaryrd}

\usepackage{tikz}
\usetikzlibrary{arrows,automata,backgrounds,calc,chains,fadings,folding,decorations.fractals,decorations.pathreplacing,fit,patterns,positioning,mindmap,shadows,shapes.geometric,shapes.symbols,through,trees,plotmarks}

%% file: macros.tex
\theoremstyle{plain}
\newtheorem{theorem}{Theorem}%[section]
\newtheorem{corollary}[theorem]{Corollary}

\newtheorem{proposition}[theorem]{Proposition}
\newtheorem{conjecture}[theorem]{Conjecture}
\theoremstyle{definition}
\newtheorem{definition}[theorem]{Definition}
\newtheorem{remark}[theorem]{Remark}
\newtheorem{example}[theorem]{Example}

\newtheorem{convention}[theorem]{Convention}

\renewcommand{\mit}{\mathit}

\newcommand{\mcl}{\mathcal}
\newcommand{\msf}{\mathsf}

\newcommand{\mbb}{\mathbb}

\newcommand{\floor}[1]{\lfloor#1\rfloor}

\newcommand{\punc}[1]{\:\text{#1}}

\newcommand{\del}[1]{}
\newcommand{\sub}[2]{#1_{#2}}
\newcommand{\super}[2]{#1^{#2}}

\newcommand{\funap}[2]{#1(#2)}

\newcommand{\juxt}[2]{#1#2}

\newcommand{\mylam}[2]{\lambda#1.#2}
\newcommand{\mymu}[2]{\mu#1.#2}

\newcommand{\infinity}{\infty}

\newcommand{\BNFis}{\mathrel{{:}{:}{=}}}
\newcommand{\BNFor}{\mathrel{|}}

\newcommand{\subtrmat}[2]{#1|_{#2}}

\newcommand{\posemp}{\epsilon}
\newcommand{\apos}{p}

\newcommand{\posconcat}[2]{#1#2}
\newcommand{\pos}{\funap{\mathcal{P}\!os}}

\newcommand{\sterm}{\mit{Ter}}
\newcommand{\term}{\funap{\sterm}}
\newcommand{\infterm}{\funap{\sterm^\infinity}}

\newcommand{\lambdaterm}{\term{\lambda}}
\newcommand{\lterm}{\lambdaterm}

\newcommand{\lamcal}{\lambda\beta}
\newcommand{\inflamcal}{\lambda^{\infinity}\beta}

\newcommand{\treeap}{\cdot}
\newcommand{\annotatednode}[2]{node [yshift=4mm] {$[#2]$} node {\hphantom{#1}\vphantom{#1}}}

\newcommand{\pairlft}{{\langle}}
\newcommand{\pairrgt}{{\rangle}}
\newcommand{\pairsep}{{,\,}}
\newcommand{\pairstr}[1]{\pairlft#1\pairrgt}
\newcommand{\pair}[2]{\pairstr{#1\pairsep#2}}

\newcommand{\tuple}{\pairstr}

\newcommand{\nat}{\mbb{N}}

\newcommand{\subst}[3]{#1[#2/#3]}

\newcommand{\length}[1]{|#1|}

\newcommand{\aclock}{%
\begin{tikzpicture}%
  \foreach \a in {0,30,...,330} {
    \draw [very thick] (\a:1.8mm) -- (\a:2.5mm);
  }
  \draw [thick] (0:0mm) -- (170:1.7mm);
  \draw [very thick] (0:0mm) -- (40:1.5mm);
\end{tikzpicture}%
}

\newcommand{\aaclock}{%
\begin{tikzpicture}[very thin]%
  \fill (0:2.5mm) arc (0:60:2.5mm) -- (60:1mm) arc (60:0:1mm) -- cycle;
  \fill (120:2.5mm) arc (120:180:2.5mm) -- (180:1mm) arc (180:120:1mm);
  \fill (240:2.5mm) arc (240:300:2.5mm) -- (300:1mm) arc (300:240:1mm);
  \fill (0,0) circle (0.5mm);
\end{tikzpicture}%
}

\newcommand{\sink}{\bot}
\newcommand{\sbohm}{\msf{BT}}
\newcommand{\bt}{\funap{\sbohm}}
\newcommand{\scbohm}{\sbohm_{\hspace*{-.4ex}\scalebox{.5}{\text{\aclock}}}\hspace*{-.1ex}}
\newcommand{\cbohm}{\funap{\scbohm}}
\newcommand{\slevi}{\msf{LLT}}
\newcommand{\sclevi}{\slevi_{\hspace*{-.4ex}\scalebox{.5}{\text{\aclock}}}\hspace*{-.1ex}}
\newcommand{\clevi}{\funap{\sclevi}}
\newcommand{\sber}{\msf{BeT}}
\newcommand{\scber}{\sber_{\hspace*{-.4ex}\scalebox{.5}{\text{\aclock}}}\hspace*{-.1ex}}
\newcommand{\cber}{\funap{\scber}}

\newcommand{\sabohm}{\sbohm_{\hspace*{-.4ex}\scalebox{.5}{\text{\aaclock}}}\hspace*{-.1ex}}
\newcommand{\abohm}{\funap{\sabohm}}

\newcommand{\salevi}{\slevi_{\hspace*{-.4ex}\scalebox{.5}{\text{\aaclock}}}\hspace*{-.1ex}}

\newcommand{\saber}{\sber_{\hspace*{-.4ex}\scalebox{.5}{\text{\aaclock}}}\hspace*{-.1ex}}

\newcommand{\relat}[2]{#1_{#2}}
\newcommand{\relev}[1]{#1_{\exists}}

\newcommand{\asclock}{\raisebox{-.4ex}{{\hspace*{-.1ex}\scalebox{.5}{\text{\aclock}}}\hspace*{.1ex}}}
\newcommand{\crel}[1]{\mathrel{#1_{\asclock}}}
\newcommand{\crelat}[2]{\mathrel{#1_{\asclock,#2}}}
\newcommand{\crelev}[1]{\mathrel{#1_{\asclock,\exists}}}
\newcommand{\creli}[1]{\mathrel{\crelat{#1}{\infty}}}

\newcommand{\treeequal}[1]{=_{#1}}

\newcommand{\annotate}[2]{[#1]#2}
\newcommand{\deannotate}[1]{\floor{#1}}

\newcommand{\smred}{{\twoheadrightarrow}}
\newcommand{\mred}{\mathrel{\smred}}

\newcommand{\redi}{\leftarrow}
\newcommand{\smredi}{{\twoheadleftarrow}}
\newcommand{\mredi}{\mathrel{\smredi}}

\newcommand{\sconv}{\sub{{=}}{\beta}}
\newcommand{\conv}{\mathrel{\sconv}}
\newcommand{\notconv}{\not\conv}
\newcommand{\nconv}{\notconv}

\newcommand{\sdefeq}{{=}}
\newcommand{\defeq}{\mathrel{\sdefeq}}

\newcommand{\redn}{\super{\to}}

\newcommand{\shred}{\to_\mit{h}}
\newcommand{\hred}{\mathrel{\shred}}
\newcommand{\hredn}[1]{\mathrel{\shred^{#1}}}
\newcommand{\hredeq}{\mathrel{\shred^{\equiv}}}
\newcommand{\smhred}{{\twoheadrightarrow_\mit{h}}}
\newcommand{\mhred}{\mathrel{\smhred}}

\newcommand{\hredat}[1]{\mathrel{\to_{\mit{h},#1}}}

\newcommand{\shredi}{{\redi_h}}
\newcommand{\hredi}{\mathrel{\shredi}}

\catcode`\@=11
\newcommand{\wash}{\mathpalette\mywash}
\newcommand{\mywash}[2]{\setbox0=\hbox{$\m@th#1{#2}$}\wd0=0pt\box0}
\catcode`\@=12
\newcommand{\arsdev}{\mathrel{\wash{\,\,\,\circ}\mathord{\longrightarrow}}}

\newcommand{\dred}{\arsdev}%{\mathrel{\protect\moverlay{\sred\cr{\hspace{-.3ex}\circ}}}}
\newcommand{\levi}{L\'{e}vy--Longo}
\newcommand{\ber}{Berarducci}
\newcommand{\bohm}{B\"{o}hm}
\newcommand{\boehm}{\bohm}

\newcommand{\sfpc}{Y}
\newcommand{\fpc}{\sub{\sfpc}}
\newcommand{\fpcC}{\fpc{0}}
\newcommand{\fpcT}{\fpc{1}}
\newcommand{\afpc}{\sfpc}
\newcommand{\awfpc}{Z}

\newcommand{\cxthole}{\Box}

\newcommand{\app}{\juxt}
\newcommand{\leftappiterate}[3]{\app{#1}{\super{#2}{{\sim}#3}}}
\newcommand{\rightappiterate}[3]{\app{\super{#1}{#2}}{#3}}

\newcommand{\nf}{\underline}

\newcommand{\posexp}[2]{#1 \times #2}    

\newcommand{\slang}{\mcl{L}}
\newcommand{\lang}{\sub{\slang}}

\newcommand{\setemp}{{\varnothing}}

\newcommand{\UN}{\msf{UN}}
\newcommand{\UNinf}{\UN^\infinity}
\newcommand{\CR}{\msf{CR}}
\newcommand{\CRinf}{\CR^\infinity}

%% file: abstract.tex
Fixed point combinators (and their generalization: looping combinators) are classic notions
belonging to the heart of $\lambda$-calculus and logic. We start with an exploration of the
structure of fixed point combinators (fpc's), vastly generalizing the well-known fact that if $Y$ is an fpc, 
$Y(SI)$ is again an fpc, generating the \boehm{} sequence of fpc's. 
Using the infinitary $\lambda$-calculus we devise infinitely many other generation schemes for fpc's.
In this way we find schemes and building blocks to construct new fpc's in a modular way. 

Having created a plethora of new fixed point combinators, the task is to prove that they are indeed new.
That is, we have to prove their $\beta$-inconvertibility. 
Known techniques via \boehm{} Trees do not apply, because all fpc's have the same \boehm{} Tree (BT). 
Therefore, we employ `clocked BT's', with annotations that convey information of the tempo in which 
the data in the BT are produced. BT's are thus enriched with an intrinsic clock behaviour, 
leading to a refined discrimination method for $\lambda$-terms. 
The corresponding equality is strictly intermediate between $\sconv$ and $=_\sbohm$, 
the equality in the classical models of $\lambda$-calculus. An analogous approach pertains to \levi{} and \ber{} trees.
Finally, we increase the discrimination power by a precision of the clock notion that we call `atomic clock'.

%% file: intro.tex
\bohm{} trees constitute a well-known method to discriminate $\lambda$-terms $M$, $N$:
if $\bt{M}$ and $\bt{N}$ are not identical, then $M$ and $N$ are $\beta$-inconvertible, $M \nconv N$.
But how do we prove $\beta$-inconvertibility of $\lambda$-terms with the same BT?
This question was raised in Scott~\cite{scott:1975} for the interesting equation $BY = BYS$ 
between terms that as Scott noted are presumably $\beta$-inconvertible, yet BT-equal ($=_{\sbohm}$).
Scott used his Induction Rule to prove that $BY =_{\sbohm} BYS$;
instead we will employ below the infinitary $\lambda$-calculus with the same effect,
but with more convenience for calculations as a direct generalization of finitary $\lambda$-calculus.
Often one can solve such a $\beta$-discrimination problem by finding a suitable invariant 
for all the $\beta$-reducts of $M$, $N$.
Below we will do this by way of preparatory example for the fpc's in the \bohm{} sequence.
But a systematic method for this discrimination problem has been lacking, 
and such a method is one of the two contributions of this paper.

%\begin{SCfigure}[1.3][h]
\begin{figure}[h]
\begin{center}
\begin{tikzpicture}[thick,node distance=13mm,inner sep=0.5mm]
  \node (BT) {$\sbohm$};
  \node (cBT) [left of=BT,yshift=-.6cm]{$\scbohm$};
  \node (aBT) [left of=cBT,yshift=-.6cm]{$\sabohm$};

  \node (LLT) [below of=BT] {$\slevi$};
  \node (cLLT) [left of=LLT,yshift=-.6cm]{$\sclevi$};
  \node (aLLT) [left of=cLLT,yshift=-.6cm]{$\salevi$};
  
  \node (BeT) [below of=LLT] {$\sber$};
  \node (cBeT) [left of=BeT,yshift=-.6cm]{$\scber$};
  \node (aBeT) [left of=cBeT,yshift=-.6cm]{$\saber$};

  \node (beta) [below of=aBeT]{$\conv$};

  \draw (BT) -- (cBT) -- (aBT);
  \draw (LLT) -- (cLLT) -- (aLLT);
  \draw (BeT) -- (cBeT) -- (aBeT);

  \draw (BT) -- (LLT) -- (BeT);
  \draw (cBT) -- (cLLT) -- (cBeT);
  \draw (aBT) -- (aLLT) -- (aBeT);
  
  \draw (aBeT) -- (beta);
\end{tikzpicture}
\vspace{-2ex}
\caption{Comparison of (atomic) clock semantics and unclocked semantics. Higher means more identifications.}
\label{fig:overivew}
\vspace{-4ex}
\end{center}
\end{figure}
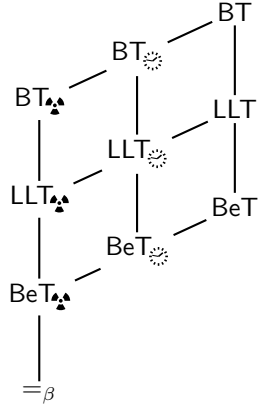
%\end{SCfigure}

Actually, the need for such a strategic method was forced upon us, by the other contribution,
because Scott's equation $BY = BYS$ turned out to be the key unlocking a plethora of new fpc's.
The new generation schemes are of the form: if $Y$ is a fpc, then $YP_1\ldots P_n$ is a fpc,
abbreviated as $Y \Rightarrow YP_1\ldots P_n$.
So $\cxthole P_1\ldots P_n$ is a `fpc-generating' vector,
and can be considered as a building block to make new fpc's. 
But are they indeed new? 
A well-known example of a (singleton)-fpc-generating vector is $\cxthole\delta$, where $\delta = SI$,
giving rise when starting from Curry's fpc to the \boehm{} sequence of fpc's.
Here another interesting equation is turning up, namely $Y = Y \delta$, 
for an arbitrary fpc $Y$, considered by Statman and Intrigila.  
In fact, it is implied by Scott's equation:
\begin{align*}
  BY = BYS \,\,\Longrightarrow\,\, BYI = BYSI \,\,\Longleftrightarrow\,\, Y = Y\delta
\end{align*}
The first equation $BY = BYS$ will yield many new fcp's, built in a modular way;
the last equation $Y = Y\delta$ addresses the question whether they are indeed new.
Finding ad hoc invariant proofs for their novelty is too cumbersome. 
But fortunately, it turns out that although the new fpc's all have the same BT, namely $\mylam{f}{f^\omega}$,
they differ in the way this BT is formed, in the `tempo of formation', where the ticks of the clock are head reduction steps.
More generally, we can discern a clock-like behaviour of BT's, that enables us to discriminate the terms in question.
However, this refined discrimination method does not work for all $\lambda$-terms; only for a class of `simple' terms,
that still is fairly extensive; it includes all fpc's that are constructed in the modular way that we present, thereby solving our novelty problem.
In fact, we gain some more ground: though our discrimination method works best 
for pairs of simple terms, it can also fruitfully be applied
to compare a simple term with a non-simple term,
and with some more effort, we can even compare
and discriminate two non-simple terms.

Even so, many pairs of fpc's cannot yet be discriminated,
because they not only have the same BT,
they also have the same clocked BT. Therefore, in a final grading up of the precision, 
we introduce `atomic clocks', where the actual position of a head reduction step is administrated. 
All this pertains not only to the BT-semantics, but also to \levi{} Trees (LLT) (or lazy trees), 
and Berarducci Trees (BeT) (or syntactic trees). 
Many problems stay open, in particular problems generalizing the equation of Statman and Intrigila, 
when arbitrary fpc's are considered --- indeed, we have only scratched the surface.

%% file: prelims.tex
To make this paper moderately self-contained,
and to fix notations, we lay out some ingredients.
For $\lambda$-calculus we refer to \cite{bare:1984} and \cite{beth:2003}.
For an introduction to \boehm{}, \ber{} and \levi{} trees, 
we refer to \cite{bare:1984,abra:ong:1993,beth:klop:vrij:2000,bare:klop:2009}.

\begin{definition}
  $\lambda$-terms are defined by the grammar:
  \begin{align*}
  M & \BNFis x \BNFor \mylam{x}M \BNFor MM
  \end{align*}
  We let $\lterm$ denote the set of $\lambda$-terms, 
  and use $M,N,\ldots$ to range over the elements of $\lterm$.
  The relation $\to_\beta$ is the compatible closure (i.e., closure under term formation)
  of the $\beta$-rule:
  \begin{align}
    (\mylam{x}{M})N & \to \subst{M}{N}{x}
    \tag{$\beta$}
  \end{align}
  where $\subst{M}{N}{x}$ denotes the result of substituting $N$ 
  for all free occurrences of $x$ in $M$.
  Furthermore, we use $\smred_\beta$ to denote the reflexive--transitive closure of $\to_\beta$.
  We write $M \conv N$ to denote that $M$ is $\beta$-convertible with $N$,
  i.e., $\sconv$ is the equivalence closure of $\to_\beta$.
  For syntactic equality (modulo renaming of bound variables), we use $\equiv$.
  We will often omit the subscript $\beta$ in $\to_\beta$ and $\smred_\beta$, 
  but not so for $\conv$, in order to reserve $=$ for definitional equality.
  
  A $\lambda$-term $M$ are called \emph{normal form} if there exists no $N$ with $M \to N$.
  We say that a term $M$ \emph{has a normal form} if it reduces to one.
  For $\lambda$-terms $M$ having a normal form we write
  $\nf{M}$ for the unique normal form $N$ with $M \mred N$ (note that
  uniqueness follows from confluence of the $\lambda$-calculus).
\end{definition}

Some commonly used combinators are:
\begin{align*}
  I & \defeq \mylam{x}{x} &
  S & \defeq \mylam{xyz}{xz(yz)} &
  B & \defeq \mylam{xyz}{x(yz)}
\end{align*}

\begin{definition}
  A \emph{position} is a sequence over $\{0,1,2\}$.
  The \emph{subterm $\subtrmat{M}{\apos}$ of $M$ at position $\apos$}
  is defined by:
  \begin{align*}
    \subtrmat{M}{\posemp} &= M
    &
    \subtrmat{(MN)}{\posconcat{1}{\apos}} 
    & = \subtrmat{M}{\apos} 
    \\
    \subtrmat{(\mylam{x}{M})}{\posconcat{0}{\apos}} 
    & = \subtrmat{M}{\apos} 
    &
    \subtrmat{(MN)}{\posconcat{2}{\apos}} 
    & = \subtrmat{N}{\apos}
  \end{align*}
  $\pos{M}$ is the set of positions $\apos$ such that $\subtrmat{M}{\apos}$ is defined.
\end{definition}

\begin{definition}
  \hfill
  \begin{enumerate}
    \item 
      A term $\afpc$ %\in\lterm$ 
      is an \emph{fpc} 
      if $\afpc x \conv x(\afpc x)$.
    \item 
      An fpc $\afpc$ is \emph{$k$-reducing} 
      if $\afpc x \redn{k} x(\afpc x)$. %\in\lterm$.
    \item 
      A term $\awfpc$ is a \emph{weak fpc (wfpc)} 
      if $\awfpc x \conv x (\awfpc' x)$
      where $\awfpc'$ is a wfpc.
  \end{enumerate}
\end{definition}
\noindent
A wfpc is alternatively defined as a term having the same B\"{o}hm tree
as an fpc, namely $\mylam{x}{x^{\omega}} \equiv \mylam{x}{x(x(x(\ldots}$\,. \\
Weak fpc's are known in foundational studies of type systems 
as \emph{looping combinators}; see, e.g., \cite{coqu:herb:1994} and \cite{geuv:wern:1994}.
\begin{example}
  Define by double recursion, $Z$ and $Z'$ such that 
  $Zx = x(Z'x)$ and $Z'x = x(Zx)$.
  Then $Z,Z'$ are both wfpc's, and $Zx = x(x(Zx))$. 
  So $Z$ delivers its output twice as fast as an ordinary fpc, but the generator flipflops.
\end{example}
As to `double recursion', \cite{klop:2007}
collects several proofs of the double fixed point theorem, 
including some in~\cite{bare:1984,smull:1985}.

\begin{definition}
  \hfill
  \begin{enumerate}
  \item 
  A \emph{head reduction step $\hred$} is a $\beta$-reduction step of the form:
  \[\mylam{x_1\ldots x_n}{(\mylam{y}{M})N N_1 \ldots N_m} 
   \to \mylam{x_1\ldots x_n}{(\subst{M}{y}{N})N_1 \ldots N_m}\]
  with $n,m \ge 0$.
  \item 
    Accordingly, a \emph{head normal form (hnf)} is a $\lambda$-term
    of the form \[\mylam{x_1}{\ldots\mylam{x_n}{y N_1 \ldots N_m}}\] with $n, m \ge 0$.
  \item
    A \emph{weak head normal form (whnf)} is 
    an hnf or an abstraction,
    that is, a whnf is a term of the form $x M_1 \ldots M_m$ or $\mylam{x}{M}$.
  \item
    A term \emph{has a (weak) hnf} if it reduces to one.
  \item
    We call a term \emph{root-stable} 
    if it does not reduce to a redex: $(\mylam{x}{M}) N$.
    A~term is called \emph{root-active} if it does not reduce
    to a root-stable term.
  \end{enumerate}
\end{definition}

\paragraph{Infinitary $\lambda$-calculus $\boldsymbol{\inflamcal}$.}
We will only use the infinitary $\lambda$-calculus $\inflamcal$ for some simple calculations 
such as $(\mylam{ab}{(ab)^\omega}) I = \mylam{b}{(I b)^\omega} = \mylam{b}{b^\omega}$.
For a proper setup of $\inflamcal$ we refer to
\cite{bera:intr:1996,kenn:klop:slee:vrie:1997,kenn:vrie:2003,bare:klop:2009}.
%Ber-Intrigila, KKSdeVRies, Terese ch.12 (Kenn-devries), Barendregt-klop 2009. 
%
In a nutshell, 
$\inflamcal$ extends finitary $\lambda$-calculus by admitting infinite $\lambda$-terms, 
the set of which is called $\infterm{\lambda}$, 
and infinite reductions (in~\cite[Ch.~12]{terese:2003} and~\cite{bare:klop:2009} 
possibly transfinitely long, in~\cite{bera:intr:1996} of length $\leq \omega$).
Limits of infinite reduction sequences are obtained by a strengthening of Cauchy-convergence, 
stipulating that the depth of contracted redexes must tend to infinity. 
%%% \^infty beta
The $\inflamcal$-calculus is not infinitary confluent ($\CRinf$), 
but still has unique infinite normal forms ($\UNinf$). 
\boehm{} Trees (BT's) without $\bot$ are infinite normal forms in 
$\inflamcal$.
But beware, the reverse does not hold, e.g.\ $\mylam{x}{(\mylam{x}{(\mylam{x}{\ldots}}}$
is an infinite normal form, but not a BT; it is in fact an LLT (\levi{} Tree, 
and also a BeT (Berarducci Tree). 
The notions BT, LLT, BeT are defined e.g.\ in~\cite{bare:klop:2009},
and in~\cite{beth:klop:vrij:2000}.
These notions are also defined in Sections~\ref{sec:clocked} and~\ref{sec:levi}, 
via their clocked versions.

\begin{definition}
  For terms $A,B$ we define
  $\leftappiterate{A}{B}{n}$ and $\rightappiterate{A}{n}{B}$:
  \begin{align*}
  \leftappiterate{A}{B}{0}    & = A 
  &&& 
  \rightappiterate{A}{0}{B}   & = B 
  \\
  \leftappiterate{A}{B}{n+1}  & = \leftappiterate{\app{A}{B}}{B}{n} 
  &&& 
  \rightappiterate{A}{n+1}{B} & = \app{A}{(\rightappiterate{A}{n}{B})}
  \end{align*}
  A context of the form $\leftappiterate{\cxthole}{B}{n}$ 
  is called a \emph{vector}.
  For the vector notation, it is to be understood that 
  term formation gets highest priority,
  i.e., 
  $\leftappiterate{AB}{C}{n} = \leftappiterate{(AB)}{C}{n}$.
\end{definition}

%\begin{definition}\label{def:appl:arith:seq}
%  \hfill
%  \begin{enumerate}
%
%    \item 
%      Let $A,B,C,\ldots$ be an infinite sequence of terms. 
%      Then we will call the sequence $A,AB,ABC,\ldots$ 
%      an \emph{applicative sequence} of terms.
%
%    \item 
%      A particular case is a sequence of the form 
%      $M,MP,MPP,\ldots$, that we will call 
%      an \emph{arithmetical sequence}. 
%
%  \end{enumerate}
%\end{definition}

%So the \boehm{} sequence is an arithmetical sequence.

%% file: boehm.tex
There are several ways to make fpc's. 
For heuristics behind the construction of Curry's fpc 
$Y_0 \defeq \mylam{f}{\omega_{f}\omega_{f}}$,
with $\omega_{f} \defeq \mylam{x}{f(xx)}$, 
and Turing's fpc $\fpcT \defeq \eta\eta$
with  $\eta \defeq \mylam{xf}{f(xxf)}$,
see~\cite{klop:2007}.
The following is an easy exercise. 
\begin{proposition}\label{prop:Y0neqY1}
  The fpc's $\fpcC$ and $\fpcT$ are not $\beta$-convertible. 
\end{proposition}

It is well-known, as observed by C.~B\"{o}hm~\cite{bohm:1963} and others, 
that the class of fpc's coincides exactly with the class of fixed points 
of the peculiar term $\delta = \mylam{ab}{b(ab)}$, convertible with $S I$. 
The notation $\delta$ is convenient for calculations 
and stems from%B.~Intrigila
~\cite{intri:1997}.

This term also attracted the attention of R.~Smullyan,
in his beautiful fable about fpc's figuring as birds 
in an enchanted forest: 
  ``An extremely interesting bird is the owl $O$ defined
    by the following condition: $Oxy = y(xy)$.'' 
  \cite%[pp.~133--134]
  {smull:1985}. 
We will return to the Owl in Remark~\ref{rem:owl} below.

Thus the term $Y \delta$ is an fpc whenever $Y$ is.
It follows that starting with $\fpc{0}$, Curry's fpc,
we have an infinite sequence of fpc's
$\fpc{0}, \fpc{0}\delta, \fpc{0}\delta\delta, \ldots, \leftappiterate{\fpc{0}}{\delta}{n}, \ldots$.
We call this sequence the \emph{\boehm{} sequence}.
We will indicate 
$\leftappiterate{\fpc{0}}{\delta}{n}$ by $\fpc{n}$. 
Note that indeed $\fpc{1}$, the notation that we had given to Turing's fpc, 
is correct in this naming convention. 
Now the question is whether all these `derived' fpc's are really new, 
in other words, whether the sequence is free of duplicates. 
This is {*}Exercise~6.8.9 in~\cite{bare:1984}. 

Note that we could also have started the sequence from another fpc than Curry's. 
Now for the sequence starting from an \emph{arbitrary} fpc $Y$, 
it is actually an open problem whether that sequence of fpc's 
$Y, Y\delta, Y\delta\delta, \ldots, \leftappiterate{Y}{\delta}{n}, \dots$ is free of repetitions. 
All we know, applying Intrigila's theorem, Theorem~\ref{thm:intrigila} below, 
is that no two consecutive fpc's 
in this sequence are convertible. But let us first consider the B\"{o}hm sequence. 

We show that the \boehm{} sequence contains no duplicates by 
determining the set of reducts of every $Y_n$.
For $\fpc{3} \equiv \eta \eta \delta \delta$,
the head reduction is displayed in 
Figure~\ref{fig:Y3:head:reduction:graph},
\begin{figure}[ht!]
  \begin{center}
  \begin{tikzpicture}[thick]
  \node (Y3) {$\eta \eta \delta \delta x$};
  \node (R) [right of=Y3,node distance=25mm] {$x(\eta \eta \delta \delta x)$};
  \draw [->,shorten >= 2mm,shorten <= 2mm] (Y3) -- (R) node [very near end,below] {$h$} node [very near end,above] {$6$};
  \draw [gray,line width=3pt] ($(R.south west)+(4mm,.0mm)$) -- ($(R.south east)+(-2mm,.0mm)$);
  \draw [->,gray,line width=3pt] ($(R.south west)!.5!(R.south east)+(1mm,.0mm)$) -- +(0mm,-.5cm) -| ($(Y3.south west)!.5!(Y3.south east)$);
  \end{tikzpicture}
  \vspace{-1ex}
  \caption{Head reduction of $\boldsymbol{\fpc{3} x}$.}
  \label{fig:Y3:head:reduction:graph}
  \vspace{-3ex}
  \end{center}
\end{figure}
but this is by no means the whole reduction graph.
For future reference we note that the head reduction diagram suggests a `clock behaviour'.

\begin{theorem}\label{thm:boehm:seq}
  The B\"{o}hm sequence contains no duplicates.
\end{theorem}
\begin{proof}
  (See also~\cite{klop:2007}.)
  We define languages $\lang{n}\subseteq\lterm$
  where $\lang{n}$ is the set of $\smred$-reducts of $Y_n$.
  For $n \geq 1$ we take $Y_n = \leftappiterate{\eta\eta}{\delta}{n-1}$,
  where $\eta \equiv \mylam{xf}{f(xxf)}$.
  \begin{align*}
  \lang{0} & \BNFis 
    \mylam{f}{\rightappiterate{f}{k}{(\omega_f\omega_f)}}
    && (k \geq 0)
  \\
  \lang{1} & \BNFis 
    \eta\eta % 
    \BNFor 
    \mylam{f}{\rightappiterate{f}{k}{(\lang{1} f)}} % 
    && (k > 0)
  \\
  \lang{n} & \BNFis 
    % case 1
    \lang{n-1}\delta 
    \BNFor 
     % case 2:
     \delta \lang{n}
    \BNFor 
    % case 3
    \mylam{b}{b^k (\lang{n}b)}
    && (n > 1, k > 0)
  \end{align*}
  Then we show that:
  \begin{enumerate}
  \item
  $\lang{n}$ is closed under \mbox{$\beta$-reduction}; and
  %i.e., if $M\in\lang{n}$ and $M \to N$, then $N\in\lang{n}$;
  \item
    $\lang{n}$ and $\lang{m}$ are disjoint, for $n \neq m$.
    %$\lang{n}\cap\lang{m}=\setemp$ for $n\ne m$.
  \end{enumerate}
  This implies that $Y_n \nconv Y_m$
  for all $n\ne m$.
  %together with Proposition~\ref{prop:Y0neqY1} yields $Y_n \nconv Y_m$  
  %for all $n\ne m$.
  
  For $n \neq m$, $n > 1$, %$\lang{n} \cap \lang{m} = \setemp$ 
  (ii) follows by counting the number of passive $\delta$'s.
  An occurrence of $\delta$ is passive if it occurs as $P \delta$ for some $P$. 
  To see that $\lang{0} \cap \lang{1} = \setemp$,
  note that if $M \in \lang{1}$ 
  is an abstraction
  then $M \equiv \mylam{f}{\rightappiterate{f}{k}{(P f)}}$
  containing a subterm $P f$ which is never the case in $\lang{0}$.

  We show (i): if $M\in\lang{n}$ and $M \to N$, then $N\in\lang{n}$.
  Using induction, we do not need to consider cases where the rewrite
  step is inside a variable of the grammar.
  We write $\lang{n}$ in terms as shorthand for
  a term $M \in \lang{n}$.

  \begin{enumerate}
  % L_0
    \item [($\lang{0}$)]
    We have $\mylam{f}{\rightappiterate{f}{k}{(\omega_f\omega_f)}} 
    \to \mylam{f}{\rightappiterate{f}{k+1}{(\omega_f\omega_f)}} \in \lang{0}$.

  % L_1
    \item [($\lang{1}$)]
    We have $\eta\eta \to \mylam{f}{f(\eta\eta f)} \in \lang{1}$,\\
    and $\mylam{f}{\rightappiterate{f}{{k}}{(\mylam{f}{\rightappiterate{f}{{\ell}}{(\lang{1} f)}} f)}}
    \to \mylam{f}{\rightappiterate{f}{{k+\ell}}{(\lang{1} f)}} \in \lang{1}$.

  % L_n
    \item [($\lang{n}$)]
    % case 1
    Case 1:
    $\mylam{f}{\rightappiterate{f}{{k}}{(\lang{1} f)}} \delta 
    \to \rightappiterate{\delta}{{k}}{(\lang{1} \delta)} \in \lang{n}$
    for $n = 2$,
    %by $k \times$ case $2$ and $1 \times$ case 1,
    and
    $(\mylam{b}{b^k(\lang{n-1}b)}) \delta \to \delta^k (\lang{n-1} \delta) \in \lang{n}$ 
    for $n > 2$.
    
    % case 2
    Case 2: $\delta \lang{n} \to \mylam{b}{b(\lang{n} b)} \in \lang{n}$.
    
    % case 3  
    Case 3: 
    $\mylam{b}{b^k(\mylam{c}{c^\ell(\lang{n} c)} b)} \to \mylam{b}{b^{k+\ell}(\lang{n} b)} \in \lang{n}$.
  \end{enumerate}
  %\vspace{-2ex}
\end{proof}

A very interesting theorem involving $\delta$ was proved by B.~Intrigila,
affirming a conjecture by R.~Statman. 

\begin{theorem}[Intrigila~\cite{intri:1997}]\label{thm:intrigila}
  There is no `double' fixed point combinator.
  That is, for no fpc $\afpc$ we have $\afpc\delta \conv \afpc$.
\end{theorem}

\begin{remark}[Smullyan's Owl $SI \conv \delta \equiv \mylam{xy}{y(xy)}$]\label{rem:owl}\mbox{}\\
  We collect some salient facts and questions.
  \begin{enumerate}
    \item
      If $Y$ is an fpc, then $Y\delta$ is an fpc \cite{bohm:1963}.
    \item
      %Smullyan~\cite{smull:1985}: 
      Let $Z$ be a wfpc. 
      Then both $\delta Z$ and $Z \delta$ are wfpc's \cite{smull:1985}.
    \item
      Call an applicative combination of $\delta$'s a \mbox{$\delta$-term}.
      In spite of $\delta$'s simplicity, not all \mbox{$\delta$-terms} 
      are strongly normalizing (SN). 
      An example of a \mbox{$\delta$-term} with infinite reduction 
      is $\delta\delta(\delta\delta)$ 
      (Johannes Waldman, Hans Zantema, personal communication, 2007).
      %\cite{wald:2007,zant:2007}.
    \item 
      Let $t$ be a non-trivial $\delta$-term, i.e., not a single $\delta$.
      Then $t$ is SN iff $t$ contains exactly one occurrence of $\delta\delta$.
      Furtermore, if $\delta$-terms $t,t'$ are SN, then they are convertible iff 
      $t,t'$ have the same length \cite{bare:klop:2009}.
    \item 
      Convertibility is decidable for $\delta$-terms \cite{stat:1989}.
    \item
      Call $\Delta = \delta^\omega$, so $\Delta \equiv \delta\Delta$.
      Then, the infinite $\lambda$-term $\Delta$ is an fpc:
      $\Delta x \equiv \delta\Delta x \mred x (\Delta x)$.
      $\Delta$ can be normalized again: $\Delta \to_{\omega} \mylam{f}{f^\omega}$.
      There are many more infinitary fpc's, e.g.\ for every $n$,
      the infinite term $\leftappiterate{(SS)^\omega}{S}{n} I$ is one,
      as will be clear from the sequel.
    \item 
      $\sbohm(\delta\delta(\delta\delta)) \equiv \sink$, $\delta\delta(\delta\delta)$ has no hnf.
      Its \ber{} tree is not trivial.
      Zantema remarked that $\delta$-terms, even infinite ones,
      such as $\Delta\Delta$, are ``top-terminating'' 
      (Zantema restricted himself to the applicative rule for $\delta$ only ---
      we expect that his observation remains valid for the \mbox{$\lambda\beta$-version}).
    \item 
      Is Intrigila's theorem also valid for wfpc's:
      for no wfpc $\awfpc$ we have $\awfpc \delta \conv \awfpc$?
  \end{enumerate}
\end{remark}

%% file: scott.tex
In \cite{scott:1975} the equation $BY=BYS$ 
is mentioned as an interesting example of
an equation not provable in $\lamcal$, while easily provable
with Scott's Induction Rule.%
\footnote{This equation is also discussed in~\cite{deza:seve:vrie:2003}.}
Scott mentions that he expects that using `methods of \boehm' 
the non-convertibility in $\lamcal$ can be established, 
but that he did not attempt a proof. 
On the other hand, with the induction rule the equality is easily established. 
We will not consider Scott's Induction Rule, but we will be working 
in the infinitary lambda calculus, $\inflamcal$.
It is readily verified that in $\inflamcal$ we have:
\[
  B Y = B Y S = \mylam{ab}{(ab)^{\omega}}
\]

\begin{proposition}\label{prop:BY:neq:BYS}
  $B\fpcC \notconv B\fpcC S$
\end{proposition}

\begin{proof}
  Postfixing the combinator $I$ yields $B\fpcC I$ and $B\fpcC SI$.
  Now $B \fpcC I \conv \fpcC$ and $B \fpcC SI \conv \fpcC(SI) = \fpcT$.
  Because $\fpcC \neq \fpcT$ (Proposition~\ref{prop:Y0neqY1}), 
  the result follows.
  \end{proof}
\noindent
In the same breath we can strengthen this non-equation to all fpc's $Y$, 
by the same calculation followed by an application of Theorem~\ref{thm:intrigila} 
stating that for no fpc $Y$ we have $Y \conv Y\delta \conv Y(SI)$.

\begin{remark}
  \hfill
  \begin{enumerate}
  \item
	The idea of postfixing an $I$ is suggested by the BT $\mylam{ab}{(ab)^{\omega}}$
	of $BY$ and $BYS$. Namely, in $\inflamcal$ we calculate:
	$(\mylam{ab}{(ab)^{\omega}}) I = \mylam{b}{(I b)^{\omega}} = \mylam{b}{b^{\omega}}$
	which is the BT of any fpc.
  \item
    Interestingly, Scott's equation $B Y = B YS$
    implies the equation of Statman and Intrigila,
    $Y = Y \delta$
    as one readily verifies, as in the proof of Proposition~\ref{prop:BY:neq:BYS}.
  \end{enumerate}
\end{remark}

Actually, the comparison between the terms $BY$ and $BYS$ has more
in store for us than just providing an example that the extension
from finitary lambda calculus $\lambda\beta$ to infinitary lambda
calculus $\inflamcal$ is not conservative. 
The BT-equality of $BY$ and $BYS$ suggests looking 
at the whole sequence
$BY$, $BYS$, $BYSS$, $BYSSS, \ldots, \leftappiterate{BY}{S}{n},\ldots$.
By the congruence property of BT-equality,
they all have the same BT $\lambda ab.(ab)^{\omega}$; 
so the terms in this sequence are not fpc's. 
But they are close to being fpc's,
for the first two terms in the sequence we already saw above that
postfixing an $I$ turns them into fpc's $\fpcC,\fpcT$. 
How about postfixing an $I$ to all the terms in the sequence, 
yielding
\[
BYI, BYSI,B YSSI, BYSSSI,\ldots,\leftappiterate{BY}{S}{n} I, \ldots\]
to which we will refer as the \emph{Scott sequence}.
Surprisingly, all these terms are fpc's.
The Scott sequence concurs with
the \boehm{} sequence of fpc's only for the first two elements, and then
splits off with different fpc's. 
But there is a second surprise. 
In showing that $\leftappiterate{BY}{S}{n}I$ is an fpc,
we find as a bonus the fpc-generating vector 
$\leftappiterate{\cxthole(SS)}{S}{n}I$
(which does preserve reducingness).
We collect the result.
\begin{theorem}\label{thm:scott:sequence}
  Let $Y$ be a $k$-reducing fpc. Then:
  \begin{enumerate}
	\item
      $\leftappiterate{BY}{S}{n}I$ is a (non-reducing) fpc,
	  for all $n \geq 0$\,; 
	\item
	  $\leftappiterate{Y(SS)}{S}{n}I$ is a $(k+3n+7)$-reducing fpc,
	  for all $n \geq 0$.
  \end{enumerate}
\end{theorem}
%
\begin{comment}
\begin{proof}
  By straightforward calculation.
  Let $Y$ be a $k$-reducing fpc.
  Then:
  \begin{align*}
    \leftappiterate{Y (S S)}{S}{n} I x
    & \redn{k}
    \leftappiterate{S S (Y (S S))}{S}{n} I x
    \\
    & \hredn{3n}
    S S (\leftappiterate{Y (S S)}{S}{n}) I x
    \\
    & \hredn{3}
    S I (\leftappiterate{Y (S S)}{S}{n} I) x
    \\
    & \hredn{3}
    I x (\leftappiterate{Y (S S)}{S}{n} I x)
    \\
    & \hredn{1}
    x (\leftappiterate{Y (S S)}{S}{n} I x)
  \end{align*}
  showing that $\leftappiterate{Y (S S)}{S}{n} I$ is a $(k+3n+7)$-reducing fpc.

  Furthermore, we have already seen that 
  $B Y I$ and 
  $B Y S I$ % \conv Y (S I) \conv Y \delta$,
  are fpcs; and we get that 
  $\leftappiterate{BY}{S}{n} I x \mred \leftappiterate{Y (S S)}{S}{(n-2)} I x$,
  for all $n \geq 2$.
  Hence, all $\leftappiterate{BY}{S}{n} I$ are non-reducing fpc's.
\end{proof}
\end{comment}

The proof of Theorem~\ref{thm:scott:sequence} is easy:
see the next example.

\begin{example}
  Let $Y$ be a $k$-reducing fpc.
  Then:
  \begin{align*}
    B Y S S S I x
    & \mhred Y (S S) S I x 
    \hredn{k} S S (Y (S S)) S I x
    \\
    & \hredn{3} S S (Y (S S) S) I x
    \hredn{3} S I (Y (S S) S I) x
    \\
    & \hredn{3} I x (Y (S S) S I x)
    \hredn{1} x (Y (S S) S I x)
  \end{align*}
  This shows that $\leftappiterate{B Y}{S}{3} I$ is a non-reducing fpc,
  and at the same time that $Y (S S) S I x$ is reducing.
\end{example}

\begin{remark}\label{rem:generating:vectors:n=2}
  Another `fpc-generating vector' is obtained as follows. 
  Start with the equation $Mab = ab(Mab)$; solutions all have the BT seen above,
  $\lambda ab.(ab)^{\omega}$. 
  For every $M$ satisfying this equation, we have that $MI$ is an fpc. 
  For: $MIx = Ix(MIx) = x(MIx)$. 
  Now we can solve the equation in different ways. 
  \begin{enumerate}
  \item
	$Mab = Y(ab)$, so $M=\lambda ab.Y(ab) = (\lambda yab.y(ab))Y = BY$, 
	as found before.
  \item
	$Mab = ab(Mab) = Sa(Ma)b$, which is obtained by solving
	$Ma = Sa(Ma)$, leading to $Ma = Y(Sa)=BYSa$, so $M = BYS$. 
	Also this solution was considered before. 
  \item
	$M = \lambda ab.ab(Mab) = (\lambda mab.ab(mab))M$,
	yielding $M = Y\varepsilon$ with $\varepsilon = \lambda abc.bc(abc)$. 
	So if $Y$ is an fpc, then $Y\varepsilon I$ is again an fpc.
  \end{enumerate}
\end{remark}

%% file: schemes.tex
The schemes mentioned in 
Theorem~\ref{thm:scott:sequence} and
Remark~\ref{rem:generating:vectors:n=2}(iii) 
for generating new fixed points from old, 
are by no means the only ones. 
There are in fact infinitely many of such schemes. 
They can be obtained analogously to the ones that we extracted
above from the equation $BY = BYS = \lambda ab.(ab)^{\omega}$, 
or the equation $Mab = ab(Mab)$. 
We only treat the case for $n=3$: 
consider the equation $Nabc = abc(Nabc)$. 
Then every solution $N$ is again a `pre-fpc', 
namely $NII$ is a fpc: $NIIx \conv IIx(NIIx) \conv x(NIIx)$. 
\begin{enumerate}
  \item
  $Nabc = Y(abc)$,
  which yields $N = (\lambda yabc.y(abc)))Y = (\lambda yabc.BBByabc)Y$.
  We obtain $N = BBBY$. 
  \item 
  $N = Y \xi$ with $\xi = \lambda nabc.abc(nabc)$, 
  yielding the fpc-generating vector $\cxthole \xi II$. 
  \item $Nabc = abc(Nabc) = S(ab)(Nab)c$.
  So we take $Nab = S(ab)(Nab)$, which yields $Nab = Y(S(ab)) = BBBY(BS)ab$.
  So $N = BBBY(BS)$, and thus we find the equation $BBBY = BBBY(BS)$,
  in analogy with the equation $BY \conv BYS$ above. 
\end{enumerate}
Also this equation spawns lots of fpc's as well as fpc-generating vectors. 
Let's abbreviate $BS$ by $A$. 
First one forms the % arithmetical 
sequence 
\begin{align*}
  BBBY,\; BBBYA,\; BBBYAA,\; BBBYAAA,\ldots
\end{align*}
These terms all have the BT $\lambda abc.abc(abc)^{\omega}$. 
They are not yet fpc's , but only `pre-fpc's'. 
But after postfixing this time $\ldots II$ we do again obtain a sequence of fpc's: 
\begin{align*}
  BBBYII,\; BBBYAII,\; BBBYAAII, \ldots
\end{align*}
Again the first two coincide with $\fpcC,\fpcT$, 
but the the series deviates not only from the \boehm{} sequence 
but also from the Scott sequence above. 
As above, the proof that a term in this sequence is indeed a fpc, 
yields a fpc-generating vector. 
Thus we find e.g.\ the following new fpc-generating schemes,
which we render in a self-explaining notation:
\begin{enumerate}
  \item $Y \Rightarrow Y(S(AI)I$
  \item $Y \Rightarrow Y(AAA)II$
  \item $Y \Rightarrow Y(AII)$ \label{item:scheme:3}
  \item $Y \Rightarrow Y(AAI)I$
  \item $Y \Rightarrow \leftappiterate{Y(AAA)}{A}{n}II$
\end{enumerate}
(Note: scheme~\eqref{item:scheme:3} came up out of the general search; 
one may recognize that it is not a new scheme, because the term $AII$ 
is actually the Owl $\delta$). 
We can derive many more of these schemes by proceeding
with solving the general equation $Na_{1}a_{2}...a_{n}=a_{1}a_{2}...a_{n}(Na_{1}a_{2}...a_{n})$,
bearing in mind the following proposition.
\begin{proposition}
  If $N$ is a term satisfying
  \[N a_1 a_2 \ldots a_n = a_1 a_2 \ldots a_n (N a_1 a_2 \ldots a_n)\]
  then $\leftappiterate{N}{I}{(n-1)}$ is an fpc.
\end{proposition}

We finally mention an fpc-generating scheme with `dummy parameters'.
\begin{enumerate}
  \item[(vi)]
  $Y \Rightarrow Y Q P_{1} \ldots P_{n}$ 
  where $P_{1},\ldots,P_{n}$ are arbitrary (dummy) terms,
  and 
  $Q = \mylam{yp_{1}...p_{n}x}{x(y p_1\ldots p_n x)}$.
\end{enumerate}

%% file: clocked.tex
As we have seen, there is vast space of fpc's
and there are many ways to derive new fpc's.
The question is whether all these fpc's are indeed new.
So we have to prove that they are not $\beta$-convertible.

For the \bohm{} sequence we did this by an ad hoc argument
based on a syntactic invariant; 
and this method works fine to establish lots of non-equations
between the alleged `new' fpc's that we constructed above.
Still, the question remains whether there are 
not more `strategic' ways of proving such inequalities. 

In this section we propose a more strategic way 
to discriminate terms with respect to $\beta$-conversion.
The idea is to extract from a $\lambda$-term more than just its $\sbohm$,
but also how the $\sbohm$ was formed; 
one could say, in what tempo, or in what rhythm.
A $\sbohm$ is formed from static pieces of information,
but these are rendered in a clock-wise fashion,
where the ticks of the internal clock are head reduction steps.

In the sequel we write $\annotate{k}{t}$ for
the term $t$ where the root is \emph{annotated with $k\in\nat$}.
Here, term formation binds stronger than annotation $\annotate{k}{}$.
For example $\annotate{k}{M N}$ stands for the term $\annotate{k}{(M N)}$
(that is, annotating the (non-displayed) application symbol in-between $M$ and $N$,
in contrast to $(\annotate{k}{M}) N$).
Moreover, for an annotated term $t$ we use $\deannotate{t}$
to denote the term obtaind from $t$ by dropping all annotations (including annotations of substerms).

\begin{definition}[Clocked \bohm{} trees]\label{def:cbohm}
  Let $t$ be a $\lambda$-term.
  The \emph{clocked \bohm{} tree $\cbohm{t}$ of $t$}
  is an annotated potentially infinite term defined as follows.
  If $t$ has no hnf, then define $\cbohm{t}$ as $\sink$.
  Otherwise,
  there is a head reduction $t \hredn{k} \mylam{x_1}{\ldots\mylam{x_n}{y M_1 \ldots M_m}}$ to hnf.
  Then we define 
  $\cbohm{t}$ as the term $\annotate{k}{\mylam{x_1}{\ldots\mylam{x_n}{y \cbohm{M_1} \ldots \cbohm{M_m}}}}$.
\end{definition}
\noindent
The (non-clocked) \boehm{} tree of a $\lambda$-term $M$
can be obtained by dropping the annotations:
$\sbohm(M) \defeq \deannotate{\cbohm{M}}$.

%We use $\mu$-terms 
%to denote infinite \boehm{} trees.
%For instance, the tree displayed on the left in
%Figure~\ref{fig:boem:y0:y1} is denoted by 
%$\cbohm{\fpc{0}f} = \mymu{T}{\annotate{1}fT}$.

\begin{figure}[ht!]
  \begin{center}
  \begin{tikzpicture}[level distance=7mm,inner sep=1mm]
    %\node (y0) {$\treelam{f}$} \annotatednode{$\treeap$}{1}
      \node  {$\treeap$} \annotatednode{$\treeap$}{2}
        child { node {$f$} }
        child { node {$\treeap$} \annotatednode{$\treeap$}{1}
          child { node {$f$} }
          child { node {$\treeap$} \annotatednode{$\treeap$}{1}
            child { node {$f$} }
            child { node {$\ddots$} \annotatednode{$\treeap$}{1}
              %child { node {\ldots} }
            }
          }
        };
      %};
      %\node [right of=y0] {$\;\; \equiv \cbohm{\fpcC}$};
  \end{tikzpicture}
  \begin{tikzpicture}[level distance=7mm,inner sep=1mm]
    %\node (y1) {$\treelam{f}$} \annotatednode{$\treeap$}{1}
      \node {$\treeap$} \annotatednode{$\treeap$}{2}
        child { node {$f$} }
        child { node {$\treeap$} \annotatednode{$\treeap$}{2}
          child { node {$f$} }
          child { node {$\treeap$} \annotatednode{$\treeap$}{2}
            child { node {$f$} }
            child { node {$\ddots$} \annotatednode{$\treeap$}{2}
              % child { node {\ldots} }
            }
          }
        };
      %};
      %\node [right of=y1] {$\;\; \equiv \cbohm{\fpcT}$};
  \end{tikzpicture}
  \caption{\mbox{Clocked \bohm{} trees of $\boldsymbol{\fpcC f}$ and $\boldsymbol{\fpcT f}$.}}
  \vspace{-4ex}
  \label{fig:boem:y0:y1}
  \end{center}
\end{figure}
Let us consider the fpc's $\fpcC$ of Curry and $\fpcT$ of Turing.
We have $\fpcC \equiv \mylam{f}{\omega_f\omega_f}$ 
where $\omega_f \equiv \mylam{x}{fxx}$, and
\begin{align*}
  \omega_f\omega_f \hredn{1} f (\omega_f\omega_f)
\end{align*}
Therefore we obtain 
\begin{align*}
  \cbohm{\fpcC f} = \annotate{2}{f \cbohm{\omega_f \omega_f}} 
  \quad\text{and}\quad
  \cbohm{\omega_f \omega_f} = \annotate{1}{f \cbohm{\omega_f \omega_f}}\punc.	
\end{align*}
%which we also write as $\cbohm{\omega_f \omega_f} = \mymu{T}{\annotate{1}{f T}}$.

For $\fpcT \equiv \eta \eta$ where $\eta \equiv \mylam{x}{\mylam{f}{f (xxf)}}$ we get:
\begin{align*}
  \fpcT f \equiv \eta \eta f \hredn{2} f (\eta \eta f)
\end{align*}
Hence, $\cbohm{\fpcT f} = \annotate{2}{f \cbohm{\fpcT f}}$.
Figure~\ref{fig:boem:y0:y1} displays the 
clocked \bohm{} trees of $\fpcC f$ (left) and $\fpcT f$ (right).

The following definition captures the well-known \boehm{} equality
of $\lambda$-terms.
\begin{definition}
  $\lambda$-terms $M$ and $N$ are \emph{$\sbohm$-equal},
  denoted $M \treeequal{\sbohm} N$,
  if $\sbohm(M) \equiv \sbohm(N)$.
\end{definition}

If $M$ and $N$ are not $\sbohm$-equal then $M \notconv N$.
More generally, if for some $F$, 
$\sbohm(M F) \not\equiv \sbohm(N F)$ then $M \notconv N$.
This method is know as \emph{\bohm-out technique}~\cite{bare:1984}.

Below, we refine this approach by comparing
the clocked \bohm{} trees $\cbohm{M}$ and $\cbohm{N}$
instead of the ordinary (non-clocked) \boehm{} trees.
In general, $\cbohm{M} \not\equiv \cbohm{N}$
does not always imply that $M \notconv N$.
Nevertheless, for a large class of $\lambda$-terms, called `simple' below,
this implication will turn out to be true.

We lift relations over natural numbers to relations over clocked \boehm{} trees.

\newcommand{\scbt}{T}
\newcommand{\cbt}{\sub{\scbt}}
\newcommand{\acbt}{\cbt{1}}
\newcommand{\bcbt}{\cbt{2}}

\begin{definition}
  Let $\acbt$ and $\bcbt$ be clocked \boehm{} trees
  with $\pos{\acbt} = \pos{\bcbt}$,
  ${\mathrel{R}} \subseteq \nat \times \nat$,
  and $\apos \in \pos{\acbt}$.

  We use $\acbt \relat{\mathrel{R}}{\apos} \bcbt$ to denote that either
  both $\subtrmat{\acbt}{\apos}$ and $\subtrmat{\bcbt}{\apos}$ are not annotated,
  or 
  both are annotated, and
  $\subtrmat{\acbt}{\apos} \equiv \annotate{k_1}{\acbt'}$
  and $\subtrmat{\bcbt}{\apos} \equiv \annotate{k_2}{\bcbt'}$
  with $k_1 \mathrel{R} k_2$.
  If $\acbt \relat{\mathrel{R}}{\apos} \bcbt$ for every $\apos \in \pos{\acbt}$, 
  we write $\acbt \mathrel{R} \bcbt$.

  %More generally, we write $\acbt \relat{\mathrel{R}}{A} \bcbt$
  %for a set of positions $A \subseteq \pos{\acbt}$
  %if $\acbt \relat{\mathrel{R}}{\apos} \bcbt$ for every $\apos \in A$.
  %We write $\acbt \mathrel{R} \bcbt$
  %for $\acbt \relat{\mathrel{R}}{\pos{\acbt}} \bcbt$.
  
  We write $\acbt \relev{\mathrel{R}} \bcbt$, and say \emph{$\mathrel{R}$ holds eventually},
  if there exists a depth level $\ell \in \nat$
  such that  $\acbt \relat{\mathrel{R}}{\apos} \bcbt$ for all positions $\apos \in \pos{\acbt}$ with $\length{p} \ge \ell$.
\end{definition}

%Next, we write $M \crel{R} N$ if $\cbohm{M} \mathrel{R} \cbohm{N}$ holds.
Next, we lift relations over clocked \bohm{} trees to $\lambda$-terms.
\begin{definition}
  Let $M$, $N$ be $\lambda$-terms, and ${\mathrel{R}} \subseteq \nat \times \nat$.

  We write $M \crel{\mathrel{R}} N$ whenever $M \treeequal{\sbohm} N$,
  and we have that $\cbohm{M} \mathrel{R} \cbohm{N}$.
  
  We write $M \creli{\mathrel{R}} N$ if $M \treeequal{\sbohm} N$,
  and for infinitely many $\apos \in \pos{\cbohm{M}}$ we have
  $\cbohm{M} \relat{\mathrel{R}}{\apos} \cbohm{N}$.
\end{definition}

% Init = Z;
% 
% Y = (\f. (\x. f (x x)) (\x. f (x x)));
% Z = Y (\y.\x. x (y D x));
% D = (\x.\y.y(x y));

In case of $M \crel{\le} N$ ($M \crel{\ge} N$) we say that 
$M$ has a \emph{faster} (\emph{slower}) clock than $N$.

\begin{proposition}\label{prop:clocks}
  Clocks are accelerated under reduction, 
  that is, ${\mred} \subseteq {\crel{\ge}}$, and slowing down under expansion.
\end{proposition}

\begin{proof}
  We proceed by an elementary diagram construction.
  %We illustrate the construction at the example of clocked \bohm{} trees.
  % The main observation is as follows.
  Whenever we have co-initial steps $M \hred M_1$ and $M \dred M_2$,
  then by orthogonal projection~\cite{terese:2003}
  there exist joining steps $M_1 \dred M'$ and $M_2 \hredeq M'$.
  Note that the head step $M \hred M_1$
  cannot be duplicated, only erased in case of an overlap.
  This leads to the elementary diagram displayed in Figure~\ref{fig:elementary:diagram}.
  \begin{figure}[ht!]
  \begin{center}
  \begin{tikzpicture}[thick,node distance=17mm]
    \node (M) {$M$};
    \node (M1) [right of=M] {$M_1$};
    \node (M2) [below of=M] {$M_2$};
    \node (M') [below of=M1] {$M'$};
    \draw [->,shorten >= 1mm] (M) -- (M1) node [pos=.95,below] {$h$};
    \draw [->] (M) -- (M2); \fill [fill=white,draw=black] ($(M)!.5!(M2)$) circle (1mm);
    \draw [->,shorten >= 2mm] (M2) -- (M') node [pos=.9,below] {$h$} node [pos=.9,above] {$\equiv$};
    \draw [->] (M1) -- (M'); \fill [fill=white,draw=black] ($(M1)!.5!(M')$) circle (1mm);
  \end{tikzpicture}
  \vspace{-2ex}
  \caption{Elementary diagram.}
  \vspace{-2ex}
  \label{fig:elementary:diagram}
  \end{center}
  \end{figure}
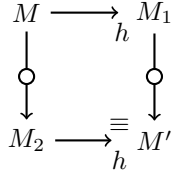

  We have ${\mred} \subseteq {\dred^*}$.
  By induction on the length of the rewrite sequence $\dred^*$
  it suffices to show that $M \dred N$ implies $M \ge N$.
  Let $M \dred N$.
  If $M$ has no hnf, then the same holds for $N$, and hence $\cbohm{M} = \bot = \cbohm{N}$.
  Therefore assume that there exists a head rewrite sequence 
  $M \hredn{k} H \equiv \mylam{x_1}{\ldots\mylam{x_n}{y M_1 \ldots M_m}}$ to hnf.
  We have \[\cbohm{M} \equiv \annotate{k}{\mylam{x_1}{\ldots\mylam{x_n}{y \cbohm{M_1} \ldots \cbohm{M_m}}}}\]

  Using the elementary diagram above ($k$ times),
  we can project $M \dred N$ over $M \hredn{k} H$,
  and obtain
  $H \dred H'$, 
  $N \hredn{\ell} H' \equiv \mylam{x_1}{\ldots\mylam{x_n}{y M_1' \ldots M_m'}}$
  with $\ell \le k$.
  Then
  $\cbohm{N} \equiv \annotate{\ell}{\mylam{x_1}{\ldots\mylam{x_n}{y \cbohm{M_1'} \ldots \cbohm{M_m'}}}}$
  and $\ell \le k$.
  Since $H \dred H'$ and $H$ is in hnf,
  we get $M_i \dred M_i'$ for every $i = 1,\ldots,m$.
  Co-recursively applying the same argument to $M_i \dred M_i'$
  yields $\cbohm{M} \ge \cbohm{N}$.
\end{proof}

While $\cbohm{M} \not\equiv \cbohm{N}$ does not imply $M \notconv N$,
the following theorem allows us to use clocked \bohm{} trees
for discriminating $\lambda$-terms:
\begin{theorem}\label{thm:general}
  Let $M$ and $N$ be $\lambda$-terms.
  Assume there exists a reduct $N'$ of $N$
  such that for no reduct $M'$ of $M$ we have
  $M' \crel{\le} N'$.
  Then $M \ne_\beta N$.
\end{theorem}
\begin{proof}
  If $M =_\beta N$ then $M =_\beta N'$ and $M \mred M' \mredi N'$
  for some $M'$. Hence $M' \crel{\le} N'$ by Proposition~\ref{prop:clocks}.
\end{proof}

\noindent
The theorem allows us to pick $N'$
while having to show that $M' \not\crel{\le} N'$ for all reducts $M'$ of $M$.
The latter condition is in general difficult to prove.
However, the theorem is of use if one of the terms 
has a manageable set of reducts, and this term happens to have slower clocks.

For a large class of $\lambda$-terms it turns out that clocks are invariant under reduction.
We call these terms `simple'.
\begin{definition}
  A redex $(\mylam{x}{A}){B}$ is called:
  \begin{enumerate}\setlength{\itemsep}{0ex}
    \item \emph{linear} if $x$ has at most one occurrence in $A$;
    \item \emph{call-by-value} if $B$ is a normal form; and
    \item \emph{simple} if it is linear or call-by-value.
  \end{enumerate}
\end{definition}
%\begin{definition}
%  A term $M$ is called \emph{call-by-value}
%  if for every reduct $M'$ of $M$ and
%  every redex $(\mylam{x}{N})N'$ contained in $M'$
%  we have that the argument $N'$ is in normal form.
%\end{definition}

The definition of simple redexes generalizes the well-known notions
of call-by-value and linear redexes. 
Next, we define simple \emph{terms}.
Intuitively, we call a term $t$ `simple' if every reduction
admitted by $t$ only contracts simple redexes.
The following definition % of simple terms
further generalises this intuition by considering
only standard reductions (to normal form):
\begin{definition}[Simple terms]
  A $\lambda$-term $t$ is \emph{simple}
  if either $t$ has no hnf, 
  or the head reduction to hnf $t \hredn{k} \mylam{x_1}{\ldots\mylam{x_n}{y M_1 \ldots M_m}}$
  contracts only simple redexes,
  and $M_1,\ldots,M_m$ are simple terms.
\end{definition}

All the fpc's in this paper are either simple or have simple reducts.
The clock of simple $\lambda$-terms is invariant under reduction,
that is, when ignoring finite prefixes of the clocked \bohm{} trees
(by reducing a term we can always make the clock values
in a finite prefix equal to $0$).
\begin{proposition}\label{prop:simple}
  Let $M$, $N$ be $\lambda$-terms such that $M$ is a simple term and $M \mred N$.
  Then $M \crelev{=} N$, that is, the clocks of $M$ and $N$ are eventually equal.
\end{proposition}
\begin{proof}
  The proof is a straightforward extension of the proof of~Proposition~\ref{prop:clocks}
  with the observation that for simple terms $M$,
  rewriting $M$ to hnf:
  \[M \hredn{k} H \equiv \mylam{x_1}{\ldots\mylam{x_n}{y M_1 \ldots M_m}}\] 
  does not duplicate redexes.
  Hence, the elementary diagrams
  are now of the form displayed in Figure~\ref{fig:elementary:diagram:simple}.
  \begin{figure}[ht!]
  \begin{center}
  \begin{tikzpicture}[thick,node distance=17mm]
    \node (M) {$M$};
    \node (M1) [right of=M] {$M_1$};
    \node (M2) [below of=M] {$M_2$};
    \node (M') [below of=M1] {$M'$};
    \draw [->,shorten >= 1mm] (M) -- (M1) node [pos=.95,below] {$h$};
    \draw [->] (M) -- (M2); 
    \draw [->] (M2) -- (M') node [midway,above] {$\varnothing$};
    \draw [->] (M1) -- (M') node [midway,right] {$\varnothing$};

    \begin{scope}[xshift=4cm]
    \node (M) {$M$};
    \node (M1) [right of=M] {$M_1$};
    \node (M2) [below of=M] {$M_2$};
    \node (M') [below of=M1] {$M'$};
    \draw [->,shorten >= 1mm] (M) -- (M1) node [pos=.95,below] {$h$};
    \draw [->] (M) -- (M2); 
    \draw [->,shorten >= 1mm] (M2) -- (M') node [pos=.9,below] {$h$};
    \draw [->] (M1) -- (M');
    \end{scope}
  \end{tikzpicture}
  \vspace{-2ex}
  \caption{Elementary diagrams for simple $\boldsymbol{M}$.}
  \vspace{-2ex}
  \label{fig:elementary:diagram:simple}
  \end{center}
  \end{figure}
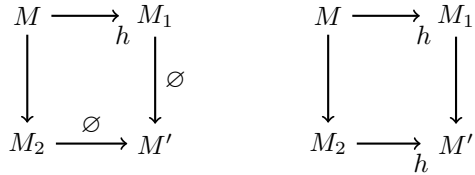
  
  That is, whenever we have co-initial steps $M \hred M_1$ and $M \to M_2$
  and $M$ is a simple term,
  then either the steps cancel each other out $M_1 \equiv M_2$ (if both are the same step),
  or they can be joined by single steps $M_1 \hred M' \redi M_2$.

  As a consequence, when projecting $M \hred M'$ over a rewrite sequence $M \to^n N$
  then either $N \to^n M'' \hredi M'$ or
  there has been cancellation and $N \to^{n-1} M'' \equiv M'$.
  Every cancellation decreases the number of steps $N \to^{n-1} M''$,
  and hence there can only finitely many cancellations.
  This implies the claim that $\cbohm{M}$ is equal to $\cbohm{N}$
  modulo a finite prefix, that is, $M \crelev{=} N$.
\end{proof}

Reduction accelerates clocks, i.e., $\smred \subseteq {\crel{{\geq}}}$.
Moreover, for simple terms the clock is invariant under reduction, see Proposition~\ref{prop:simple}.
Hence if a term $M$ has a simple reduct $N$, then $N$ has the fastest clock reachable from $M$ 
modulo a finite prefix. This justifies the following convention.
\begin{convention}
  The \emph{(minimal) clock} of a $\lambda$-term $M$ with a simple reduct $N$
  is $\cbohm{N}$, the clocked BT of $N$.
\end{convention}

For simple terms we obtain the following theorem:
\begin{theorem}\label{thm:simple}
  Let $M$ and $N$ be $\lambda$-terms.
  If there exists a reduct $N'$ of $N$
  and a simple reduct $M'$ of $M$ such that
  $M' \creli{>} N'$, then $M \ne_\beta N$.
\end{theorem}
\begin{proof}
  Assume $M =_\beta N$. Then $M' = N'$ and by confluence they have a common reduct
  $M' \mred M'' \mredi N'$.
  Since $M'$ is simple, we have $M' \crelev{=} M''$ by Proposition~\ref{prop:simple}.
%   , that is,  $\cbohm{M'}$ and $\cbohm{M''}$ differ only for in a finite prefix.
  Hence, $M' \creli{>} N'$ implies $M'' \creli{>} N'$,
  which would contradict $M'' \crel{\le} N'$ obtained by Proposition~\ref{prop:clocks}.
\end{proof}
Theorem~\ref{thm:simple} significantly reduces 
the proof obligation in comparison to Theorem~\ref{thm:general}.
We can pick any simple reduct $M'$ of $M$,
instead of having to reason about all reducts $M'$ of $M$.
For the case that both $M$ and $N$ are simple, there is no need to look for reducts:
\begin{proposition}\label{cor:simple:simple}
  For simple terms $M$ and $N$,
  $M \creli{\ne} N$ implies $M \ne_\beta N$.
\end{proposition}
\begin{proof}
  Assume $M = N$ then $M \mred O \mredi N$ for a common reduct $O$.
  Then $M \crelev{=} O \crelev{=} N$ by Proposition~\ref{prop:simple}. 
  Hence $M \crelev{=} N$ which contradicts $M \creli{\ne} N$.
\end{proof}

\begin{example}\label{ex:boehm:seq}
  Let $n \geq 2$.
  We compute the clocks of the fpc's 
  $\fpc{n}$ of the \boehm{} sequence.
  %$\omega_f = \mylam{x}{f(xx)}$. 
  We first reduce $\fpc{n} \defeq \leftappiterate{\fpcC}{\delta}{n}$
  with $\fpcC \defeq \mylam{f}{\omega_{f}\omega_{f}}$
  and $\omega_f \defeq \mylam{x}{f(xx)}$ to a simple term:
  \begin{align*}
    \fpc{n} x
    \to \leftappiterate{\omega_\delta \, \omega_\delta}{\delta}{n-1} x
    \mred \leftappiterate{\nf{\omega_\delta} \, \nf{\omega_\delta}}{\delta}{n-1} x
  \end{align*}
  where $\nf{\omega_\delta} = \mylam{ab}{b(aab)}$.
  We compute the clock:
  \begin{align*}
    \leftappiterate{\nf{\omega_\delta} \, \nf{\omega_\delta}}{\delta}{n-1} x
    & \hredn{2} \leftappiterate{\delta (\nf{\omega_\delta} \, \nf{\omega_\delta} \delta)}{\delta}{n-2} x
    \\
    & \hredn{2(n-2)} \delta (\leftappiterate{\nf{\omega_\delta} \, \nf{\omega_\delta}}{\delta}{n-1}) x
    \\
    & \hredn{2} x (\leftappiterate{\nf{\omega_\delta} \, \nf{\omega_\delta}}{\delta}{n-1} x) 
  \end{align*}
  We find 
  $\cbohm{\leftappiterate{\nf{\omega_\delta} \, \nf{\omega_\delta}}{\delta}{n-1} x} 
   = \annotate{2n}{(x \cbohm{\leftappiterate{\nf{\omega_\delta} \, \nf{\omega_\delta}}{\delta}{n-1} x})}$.
  Hence, for $n \geq 2$ the clock of $\fpc{n}$ is $2n$.
\end{example}
\noindent
By Theorem~\ref{thm:simple}, Example~\ref{ex:boehm:seq}
and Figure~\ref{fig:boem:y0:y1} we obtain an alternative 
proof for Theorem~\ref{thm:boehm:seq}: the \boehm{} sequence contains no duplicates.

\begin{example}\label{ex:scott:seq}
  Let $n \geq 2$.
  We compute the clocks of the fpc's 
  $U_n \defeq \leftappiterate{BY}{S}{n} I$
  of the Scott sequence; so where $Y = \fpcC$.
  We first reduce $U_n$ to a simple term:
  \begin{align*}
    U_n x 
    & \mred \leftappiterate{Y (S S)}{S}{(n-2)} I x \\
    & \to \leftappiterate{\omega_{SS}\,\omega_{SS}}{S}{(n-2)} I x \\
    & \mred \leftappiterate{\nf{\omega_{SS}}\,\nf{\omega_{SS}}}{S}{(n-2)} I x
  \end{align*}
  where $\nf{\omega_{SS}} = \mylam{abc}{bc(aabc)}$.
  We abbreviate $\theta \equiv \nf{\omega_{SS}}$.
  Then we compute the clocks for $n = 2$, $n = 3$, and $n > 3$:
  \begin{align*}
    \theta\theta I x 
    & \hredn{3} I x (\theta\theta I x) 
    \hredn{1} x (\theta\theta I x) 
    \\
    \theta\theta S I x 
    & \hredn{3} {S I(\theta\theta S I)} x
    \hredn{4} {x (\theta\theta S I x)} 
    \\
    \leftappiterate{\theta\theta}{S}{(n-2)} I x
    &\hredn{3} \leftappiterate{ S S (\theta\theta S S)}{S}{(n-4)} I x
    \\
    &\hredn{3(n-4)} S S (\leftappiterate{ \theta \theta S S }{S}{(n-4)}) I x
    \\
    &\hredn{3} S I (\leftappiterate{ \theta \theta }{S}{(n-2)} I) x
    \\
    &\hredn{4} x (\leftappiterate{ \theta \theta }{S}{(n-2)} I x)
  \end{align*}
  respectively.
  For all three cases, we find: %the clocked \boehm{} tree:
  \[\cbohm{\leftappiterate{\theta\theta}{S}{(n-2)} I x} 
   = \annotate{3(n-2)}{(x \cbohm{\leftappiterate{ \theta \theta }{S}{(n-2)} I x})}\]
\end{example}
\noindent
Using Theorem~\ref{thm:simple} we infer from Example~\ref{ex:scott:seq}
and Figure~\ref{fig:boem:y0:y1} (recall that $U_0 \conv Y_0$ and $U_1 \conv Y_1$):
\begin{corollary}
  The Scott sequence contains no duplicates.  
\end{corollary}

Plotkin~\cite{plot:2007} asked:
Is there an fpc $Y$ such that
\begin{align*}
  A_Y \equiv Y(\mylam{z}{fzz}) \conv Y(\mylam{x}{Y(\mylam{y}{fxy})}) \equiv B_Y 
\end{align*}
or in other notation: $\mymu{z}{fzz} \conv \mymu{z}{fzz}$,
where 
%with the definition 
$\mymu{x}{M(x)} = Y(\mylam{x}{M(x)})$.
The terms $A_Y$ and $B_Y$ have the same \boehm{} tree,
%The \boehm{} tree of both $A_Y$ and $B_Y$ is 
which is the solution of $T = f T T$.

The terms $A_Y$ and $B_Y$ are not simple.
An extension of our clock method can be given 
which restricts the clock comparison to single paths in the clocked \boehm{} tree
along which there is no duplication of redexes.
We leave this extension to future work.
Using this extension would allow us to settle 
the question in the negative for all simple fpc's.
For Turing's fpc $\fpcT$ this is seen by computing 
the clocked BT's of $A_{\fpcT}$ and $B_{\fpcT}$.
Recall $\fpcT \equiv \eta\eta$ with $\eta \equiv \mylam{xf}{f(xxf)}$.
\begin{align*}
  A_{\fpcT} 
  &\equiv \eta\eta(\mylam{z}{fzz})
  \hredn{2} (\mylam{z}{fzz}) A_{\fpcT} % (\eta\eta(\mylam{z}{fzz}))
  \\
  & \hredn{1} f A_{\fpcT} A_{\fpcT} % (\eta\eta(\mylam{z}{fzz})) (\eta\eta(\mylam{z}{fzz}))
  \\
  B_{\fpcT} 
  &\equiv \eta\eta(\mylam{x}{\eta\eta(\mylam{y}{fxy})})
  \\
  &\hredn{2} 
  (\mylam{x}{\eta\eta(\mylam{y}{fxy})}) B_{\fpcT} %%(\eta \eta (\mylam{x}{\eta\eta(\mylam{y}{fxy})}))
  \\
  &\hredn{1} 
  \eta\eta(\mylam{y}{f B_{\fpcT} y})
  \\
  &\hredn{2}
  (\mylam{y}{f B_{\fpcT} y})(\eta\eta(\mylam{y}{f B_{\fpcT} y}))
  \\
  &\hredn{1}
  f B_{\fpcT} (\eta\eta(\mylam{y}{f B_{\fpcT} y}))
\end{align*}
Note that for $B_{\fpcT}$ developing the left branch takes six steps,
whereas the right only needs three.
The clocked BT's for $A_{\fpcT}$ and $B_{\fpcT}$ are depicted in Figure~\ref{fig:plotkin}
using hnf-notation (see~\cite{bare:1984} or~\cite{bare:klop:2009}).

\begin{figure}[ht!]
\begin{center}
  \begin{tikzpicture}[level distance=7mm,inner sep=0.5mm,
                      level 1/.style={sibling distance=18mm},
                      level 2/.style={sibling distance=9mm},
                      level 3/.style={sibling distance=4.5mm},
                      level 4/.style={sibling distance=2.25mm}
                     ]
    \node  {$f$} \annotatednode{$f$}{3}
      child { node {$f$} \annotatednode{$f$}{3}
        child { node {$f$} \annotatednode{$f$}{3}
          child { node {$\ldots$} }
          child { node {$\ldots$} }
        }
        child { node {$f$} \annotatednode{$f$}{3}
          child { node {$\ldots$} }
          child { node {$\ldots$} }
        }
        %edge from parent node[above left] {$\annotate{3}$}
      }
      child { node {$f$} \annotatednode{$f$}{3}
        child { node {$f$} \annotatednode{$f$}{3}
          child { node {$\ldots$} }
          child { node {$\ldots$} }
        }
        child { node {$f$} \annotatednode{$f$}{3}
          child { node {$\ldots$} }
          child { node {$\ldots$} }
        }
      };
  \end{tikzpicture}\quad
  \begin{tikzpicture}[level distance=7mm,inner sep=0.5mm,
                      level 1/.style={sibling distance=18mm},
                      level 2/.style={sibling distance=9mm},
                      level 3/.style={sibling distance=4.5mm},
                      level 4/.style={sibling distance=2.25mm}
                     ]
    \node  {$f$} \annotatednode{$f$}{6}
      child { node {$f$} \annotatednode{$f$}{6}
        child { node {$f$} \annotatednode{$f$}{6}
          child { node {$\ldots$} }
          child { node {$\ldots$} }
        }
        child { node {$f$} \annotatednode{$f$}{3}
          child { node {$\ldots$} }
          child { node {$\ldots$} }
        }
        %edge from parent node[above left] {$\annotate{3}$}
      }
      child { node {$f$} \annotatednode{$f$}{3}
        child { node {$f$} \annotatednode{$f$}{6}
          child { node {$\ldots$} }
          child { node {$\ldots$} }
        }
        child { node {$f$} \annotatednode{$f$}{3}
          child { node {$\ldots$} }
          child { node {$\ldots$} }
        }
      };
  \end{tikzpicture}
\caption{Clocked BT's for $\boldsymbol{A_{\fpcT}}$ and $\boldsymbol{B_{\fpcT}}$.}
\label{fig:plotkin}
\end{center}
\end{figure}

We conjecture that for no fpc $Y$, $A_Y \conv B_Y$;
maybe this requires an extension of the proof in~\cite{intri:1997}.

%% file: atomic.tex
We have introduced clocked \bohm{} trees for discriminating $\lambda$-terms.
In this section, we refine the clocks to measure not only the \emph{number} of head
steps, but, in addition, the \emph{position} of each of these steps.
We call these clocks `atomic'.

Let us consider a motivating example.
We discriminate $Y_2$ and $U_2$.
First, we reduce both terms to simple reducts:
\begin{align*}
  Y_2 x &\equiv \fpcC \delta \delta x \mred \xi \xi \delta x && \text{where $\xi = \mylam{ab}{b(aab)}$}\\
  U_2 x &\equiv \fpcC (S S) I x \mred \theta \theta I x && \text{where $\theta = \mylam{abc}{bc(aabc)}$}
\end{align*}
We compute the atomic clocks of these simple reducts:
\begin{align*}
  \xi \xi \delta x
  &\hredat{11} (\mylam{b}{b(\xi\xi b)}) \delta x
  \hredat{1} \delta (\xi\xi \delta) x\\
  &\hredat{1} (\mylam{b}{b (\xi\xi \delta b)}) x
  \hredat{\posemp} x (\xi\xi \delta x)
  \\
  \theta \theta I x
  &\hredat{11} (\mylam{bc}{bc(\theta\theta bc)}) I x
  \hredat{1} (\mylam{c}{Ic(\theta\theta Ic)}) x\\
  &\hredat{\posemp} I x(\theta\theta I x)
  \hredat{1} x(\theta\theta I x)
\end{align*}
Both terms have the clocked \bohm{} tree $T \equiv \annotate{4} (x T)$.
Thus the method from the previous section is not applicable.

However, with atomic clocks we have:
\begin{align*}
  \abohm{\xi \xi \delta x} &= \annotate{11,1,1,\posemp} (x \abohm{\xi \xi \delta x})\\
  \abohm{\theta \theta I x} &= \annotate{11,1,\posemp,1} (x \abohm{\theta \theta I x})
\end{align*}
which allows us to discriminate the terms. Hence $Y_2 \ne U_2$
(by Corollary~\ref{cor:simple:simple} which generalises to the setting of atomic BT's).
Note that the (non-atomic) clocked BT's can be obtained by
taking the length of the lists of positions.

For lists $\vec{p}, \vec{q}$ of positions, we write $\vec{p} \cdot \vec{q}$
for concatenating $\vec{p}$ to $\vec{q}$.
We write $\hredat{\tuple{p_1,\ldots,p_n}}$ for the rewrite sequence
$\hredat{p_1} \cdots \hredat{p_n}$ consisting of steps at position $p_1$,\ldots,$p_n$.

\begin{definition}[Atomic clock \bohm{} trees]\label{def:abohm}
  Let $t \in \lterm$.
  The \emph{atomic clock \bohm{} tree $\abohm{t}$ of $t$}
  is an annotated infinite term defined as follows.
  If $t$ has no hnf, then define $\abohm{t}$ as $\sink$.
  Otherwise,
  there is a head reduction 
  \begin{align*}
    t \hredat{p_1} \cdots \hredat{p_k} \mylam{x_1}{\ldots\mylam{x_n}{y M_1 \ldots M_m}}
  \end{align*}
  of length $k$ to hnf.
  Then we define 
  \[
    \abohm{t}
    = \annotate{\tuple{p_1,\ldots,p_k}}{\mylam{x_1}{\ldots\mylam{x_n}{y \abohm{M_1} \ldots \abohm{M_m}}}}
  \]
\end{definition}

The theory developed for (non-atomic) BT's generalises 
as follows to the atomic trees.
For lists of positions $\vec{p}, \vec{q}$ we define
$\vec{p} \ge \vec{q}$ whenever $\vec{q}$ is a subsequence of $\vec{p}$,
and
$\vec{p} > \vec{q}$ if additionally $\vec{p} \ne \vec{q}$.
Here $\tuple{a_1,\ldots,a_n}$ is a subsequence of $\tuple{b_1,\ldots,b_m}$
if there exist indexes $i_1 < i_2 < \ldots < i_n$ such that
$\tuple{a_1,\ldots,a_n} = \tuple{b_{i_1},\ldots,b_{i_n}}$.

Using this notation for comparing the atomic annotations
(lists of positions),
Proposition~\ref{prop:clocks}, Theorem~\ref{thm:general},
Proposition~\ref{prop:simple}, Theorem~\ref{thm:simple}, and Corollary~\ref{cor:simple:simple}
remain valid
(using basically the same proofs).

\begin{proposition}\label{prop:scott:free}
  Let $B_n = \leftappiterate{\cxthole(SS)}{S}{n}I$ 
  the fpc-generating vectors from Theorem~\ref{thm:scott:sequence}.
  For $n_1,\ldots,n_k \in \nat$ we define
  $Y^{\tuple{n_1,\ldots,n_k}} = \fpcC B_{n_1} \ldots B_{n_k}$.
  We prove that all these fpc's are inconvertible, that is,
  $\vec{n} \ne \vec{m}$ implies $Y^{\vec{n}} \nconv Y^{\vec{m}}$.
\end{proposition}

This proposition cannot be proved using (non-atomic) clocks,
as for example: $\cbohm{Y^{\pair{n}{m}}} = \cbohm{Y^{\pair{m}{n}}}$.
We introduce some auxiliary notations.
Let $B'_n = \leftappiterate{\cxthole}{S}{n}I$,
and define $\nf{B_n} = \leftappiterate{\cxthole\nf{(SS)}}{S}{n}I$
where $\nf{SS} = \mylam{abc}{bc(abc)}$.
For $\vec{p} = \tuple{1^{m_1},\ldots,1^{m_k}}$ a list of positions
define $\posexp{\vec{p}}{n} = \vec{p} \cdot (\posexp{\tuple{1^{m_1-1},\ldots,1^{m_k-1}}}{(n-1)})$
for $n > 1$ and $\posexp{\vec{p}}{1} = \vec{p}$.

\begin{proof}
%  We employ notations from Example~\ref{ex:scott:seq}.
  Let $n_1,\ldots,n_k \in \nat$ and $Y \equiv \theta \theta B'_{n_1} \nf{B_{n_2}} \ldots \nf{B_{n_k}}$
  where $\theta = \mylam{abc}{bc(aabc)}$, then
  $Y^{\tuple{n_1,\ldots,n_k}} x \mred Y x$ (for $k \ge 1$) where $Yx$ is a simple term.
  Apart from the initial and final steps, the rewrite sequence $Yx \mhred x(Yx)$ 
  is composed of $k$ subsequences of the form:
  \begin{align*}
     &S I (\theta\theta \ldots) \nf{B_n} V_m \equiv
    S I \leftappiterate{ (\theta\theta \ldots) \nf{(S S)}}{S}{n} I V_m\\
    &\hredat{\posexp{1^{n+m+3}}{3}}
    I \nf{(SS)} \leftappiterate{ (\theta\theta \ldots \nf{(S S)})}{S}{n} I V_m\\
    &\hredat{1^{n+m+2}}
    \nf{SS} \leftappiterate{ (\theta\theta \ldots \nf{(S S)})}{S}{n} I V_m\\
    &\hredat{\posexp{1^{n+m+1}}{3}}
    SS \leftappiterate{ (\theta\theta \ldots \nf{(S S)})}{S}{n-2} I V_m\\
    &\hredat{\posexp{(\posexp{1^{n+m}}{3})}{(n-1)}}
    S I (\theta\theta\ldots \nf{B_n}) V_m
  \end{align*}
  for every $B_n$ with $n \ge 2$,
  and vector $V_m$ of length $m$.
  
  For every $B_n$ there is exactly one 
  occurrence of four consecutive steps at `decreasing' positions
  $1^{n+m+2}$, $1^{n+m+1}$, $1^{n+m}$, $1^{n+m-1}$ (btw, this also holds for $n < 2$).
  Hence, from the distance between these occurrences
  we can reconstruct the vector $\tuple{n_1,\ldots,n_k}$.
  This shows that $\vec{n} \ne \vec{m}$ implies that $\abohm{Y^{\vec{n}}} \relev{=} \abohm{Y^{\vec{m}}}$ is false,
  and hence we conclude $Y^{\vec{n}} \nconv Y^{\vec{m}}$ by Corollary~\ref{cor:simple:simple}.
\end{proof}

%% file: levi.tex
In fact, there are three main semantics for the $\lambda$-calculus:
$\sbohm$, $\slevi$, and $\sber$ 
(see \cite{abra:ong:1993,bera:intr:1996,beth:klop:vrij:2000,kenn:vrie:2003,bare:klop:2009}).
The notions from the previous section generalize directly to $\slevi$ and $\sber$ semantics.

\begin{definition}[Clocked \levi{} trees]
  Let $t$ be a $\lambda$-term.
  The \emph{clocked \levi{} tree $\clevi{t}$ of $t$}
  is an annotated potentially infinite term defined as follows.
  If $t$ has no whnf, then define $\clevi{t}$ as $\sink$.
  Otherwise,
  there exists a head rewrite sequence 
  $t \hredn{k} \mylam{x}{M}$ or $t \hredn{k} x M_1 \ldots M_m$ to whnf.
  In this case, we define
  $\clevi{t}$ as 
  $\annotate{k}{\mylam{x}{\clevi{M}}}$ or
  $\annotate{k}{x \clevi{M_1} \ldots \clevi{M_m}}$, respectively.
\end{definition}

\begin{definition}[Clocked \ber{} trees]
  Let $t$ be a $\lambda$-term.
  The \emph{clocked \ber{} tree $\cber{t}$ of $t$}
  is an annotated potentially infinite term defined as follows.
  If $t$ is root-active, let $\cber{t} \equiv \sink$.
  If $t \hredn{k} s$ rewrites to a root-stable term $s \equiv x$, $s \equiv \mylam{x}{M}$ or $s \equiv M N$,
  then define $\cber{t}$ as
  $\annotate{k}{x}$, $\annotate{k}{\mylam{x}{\cber{M}}}$ or $\annotate{k}{\cber{M} \cber{N}}$, respectively.
\end{definition}

\begin{example}
  We consider the terms
  $A \equiv a a$ with $a \equiv \mylam{x}{\mylam{y}{xx}}$
  and
  $B \equiv b b$ with $b \equiv \mylam{x}{\mylam{y}{\mylam{z}{xx}}}$.
  We have:
  \begin{align*}
	\clevi{A} &\equiv \annotate{1}{\mylam{y}{\clevi{A}}}\\
	\clevi{B} &\equiv \annotate{1}{\mylam{y}{\mylam{z}{\clevi{B}}}}
  \end{align*}
  Thus, in $\clevi{A}$ every $\lambda$ requires one head reduction step
  whereas in $\clevi{B}$ every second $\lambda$ is obtained for `free' (that is, in $0$ steps).
  
  We remark that $A$ and $B$ cannot be distinguished
  in the \bohm{} tree semantics since $\cbohm{A} \equiv \cbohm{B} \equiv \bot$.
  In the \bohm{} tree semantics, a term is meaningful only if it has a hnf.
  The \levi{} semantics weakens this condition to whnf's,
  and thereby allows more terms to be distinguished.
  The \ber{} tree semantics is a further weakening
  where only root-active terms are discarded as meaningless.
\end{example}

%% file: conclusion.tex
We conclude with an encompassing conjecture.
\begin{conjecture}
  Building fpc's with fcp-generating vectors is a free construction, 
  that is, there are no non-trivial identifications.
\end{conjecture}
A first step is found in Intrigila's theorem $Y \delta \nconv Y$,
for any fpc $Y$.
A second step is that the \boehm{} sequence is duplicate-free.
A third step is found in our proof that the Scott sequence is duplicate-free, 
and Proposition~\ref{prop:scott:free}, which states that 
there are no identifications when starting the construction with $\fpcC$.

Other parts of the conjecture are as follows.
Let $Y,Y'$ fpc's and $B_1 \ldots B_n$, $C_1 \ldots C_k$ be fpc-generating vectors.
\begin{enumerate}
  \item
  $Y \delta \conv Y' \delta$ iff $Y \conv Y'$.
  
  \item
  $Y B_1 \ldots B_n  \conv  Y' B_1 \ldots B_n$ iff $Y = Y'$.
  
  \item
  \mbox{$Y B_1 \ldots B_n  \nconv Y C_1 \ldots C_k$ if %$\vec{B} \not\equiv\vec{C}$.
  $B_1 \ldots B_n \not\equiv C_1 \ldots C_k$}.
\end{enumerate}

For general fpc's $Y$, $Y'$ these conjectures may be beyond current techniques, 
but for the well-known fpc's of Curry and Turing, and the fpc-generating 
vectors introduced here, including their versions for $n > 3$, 
these problems are tractable.